\documentclass[aps, pra, onecolumn, 10pt, tightenlines, notitlepage, longbibliography, superscriptaddress]{revtex4-1}

\usepackage[margin=3cm]{geometry}

\usepackage{amsmath, amsthm, amsfonts, amssymb, bm, bbm}
\usepackage{graphicx}
\usepackage{tikz}
\usepackage{hyperref, cleveref}
\hypersetup{
	colorlinks,
	linkcolor={blue},
	citecolor={blue},
	urlcolor={blue}
}
\usepackage{verbatim}

\usepackage{pgfplots}

\theoremstyle{plain}
\newtheorem{theorem}{Theorem}
\newtheorem{lemma}[theorem]{Lemma}
\newtheorem{proposition}[theorem]{Proposition}

\theoremstyle{definition}
\newtheorem{definition}{Definition}

\DeclareMathOperator{\tr}{Tr}

\newcommand{\bb}{\mathbb}
\newcommand{\mc}{\mathcal}

\newcommand{\norm}[1]{\lVert #1 \rVert}

\newcommand{\tn}[1]{\ensuremath{^{\otimes #1}}}

\newcommand{\ket}[1]{\ensuremath{|#1\rangle}}
\newcommand{\bra}[1]{\ensuremath{\langle#1|}}
\newcommand{\dket}[1]{\ensuremath{|#1\rangle\!\rangle}}
\newcommand{\dbra}[1]{\ensuremath{\langle\!\langle#1|}}

\newcommand{\drangle}{\ensuremath{\rangle\!\rangle}}

\makeatletter
\def\ketbra#1{\def\tempa{#1}\futurelet\next\ketbra@i}
\def\ketbra@i{\ifx\next\bgroup\expandafter\ketbra@ii\else\expandafter\ketbra@end\fi}
\def\ketbra@ii#1{\ket{\tempa}\!\bra{#1}}
\def\ketbra@end{\ket{\tempa}\!\bra{\tempa}}
\makeatother

\makeatletter
\def\dketbra#1{\def\tempa{#1}\futurelet\next\dketbra@i}
\def\dketbra@i{\ifx\next\bgroup\expandafter\dketbra@ii\else\expandafter\dketbra@end\fi}
\def\dketbra@ii#1{\dket{\tempa}\!\dbra{#1}}
\def\dketbra@end{\dket{\tempa}\!\dbra{\tempa}}
\makeatother

\makeatletter
\def\dbraket#1{\def\tempa{#1}\futurelet\next\dbraket@i}
\def\dbraket@i{\ifx\next\bgroup\expandafter\dbraket@ii\else\expandafter\dbraket@end\fi}
\def\dbraket@ii#1{\dbra{\tempa}#1 \drangle}
\def\dbraket@end{\dbra{\tempa}\tempa\drangle}
\makeatother

\begin{document}
\title{Classical vs. quantum satisfiability in linear constraint systems modulo an integer}
\author{Hammam Qassim \& Joel J. Wallman}
\date{\today}

\begin{abstract}
A system of linear constraints can be unsatisfiable and yet admit a solution in the form of quantum observables whose correlated outcomes satisfy the constraints. Recently, it has been claimed that such a satisfiability gap can be demonstrated using tensor products of generalized Pauli observables in odd dimensions.
We provide an explicit proof that no quantum-classical satisfiability gap in any linear constraint system can be achieved using these observables.
We prove a few other results for linear constraint systems modulo $d>2$. We show that a characterization of the existence of quantum solutions when $d$ is prime, due to Cleve et al, holds with a small modification for arbitrary $d$. We identify a key property of some linear constraint systems, called phase-commutation, and give a no-go theorem for the existence of quantum solutions to constraint systems for odd $d$ whenever phase-commutation is present. As a consequence, all natural generalizations of the Peres-Mermin magic square and pentagram to odd prime $d$ do not exhibit a satisfiability gap.
\end{abstract}
\maketitle

\section{Introduction}
The study of linear constraint systems (LCS) originated with the simple proofs of contextuality of Peres and Mermin \cite{peres1991two,Mermin1993}.
These proofs are based on inconsistent systems of linear equations over the binary field that nevertheless admit solutions if we allow for non-commutative versions of the variables, i.e. binary quantum observables.
A gap between classical and quantum satisfiability in this setting is a proof of contextuality \cite{kochen1975problem}.
In addition, such a gap can be used to define a non-local game which can be won with certainty by quantum players but not by classical players, i.e. a \textit{pseudo-telepathic} game \cite{brassard2005quantum}. This link has turned out to be immensely useful in the study of quantum non-locality \cite{arkhipov2012extending,ji2013binary,cleve2014,cleve2017perfect,slofstra2019tsirelson,slofstra2019set,dykema2019non}.
In particular, games constructed from LCSs are suitable for studying the separation between different models of quantum correlations, such as the tensor-product and the commuting-operator models. 
Difficult questions on the separation of these models can be cast as questions on the quantum satisfiability of binary LCSs, an approach which has recently culminated in a number of major results \cite{cleve2017perfect,slofstra2019tsirelson,slofstra2019set}.
Further, non-local games constructed in this way can be used for devising robust self-testing protocols \cite{coladangelo2017robust}.

Quantum satisfiability of LCSs over $\bb{Z}_2$ has been studied in detail, see for example \cite{arkhipov2012extending,ji2013binary,cleve2014,cleve2017perfect}. 
The purpose of this paper is to shed light on the non-binary case, that is, systems of linear equations over $\bb{Z}_d$, where $d>2$. This case is generally less well-understood, although it was considered to varying extents in Refs. \cite{cleve2017perfect,coladangelo2017robust,okay2019homotopical}.
For $d=2$, the canonical examples of a satisfiability gap rely on the commutation relations of qubit Pauli observables. 
Pauli observables can be generalized to higher dimensions in a number of different ways. 
The Weyl-Heisenberg group is arguably the most natural generalization, since it preserves the commutativity structure and error-correction properties of the qubit Pauli group \cite{gottesman2001encoding}. 
Elements of this group are commonly referred to as generalized Pauli operators, and are natural candidates for constructing quantum solutions for LCS over $\bb{Z}_d$ for $d>2$.
Indeed, this was the approach taken in an early version of \cite{coladangelo2017robust}, where robust self-testing routines for maximally entangled qudit states were claimed to exist using generalized Pauli observables in odd dimensions.

Unfortunately the claim in the early version of \cite{coladangelo2017robust} is false.  
In this paper, we provide a concise proof that a satisfiability gap cannot be established using tensor products of generalized Pauli observables in odd dimension.
This fact follows from the non-negative discrete Wigner function for Hilbert spaces of odd dimension \cite{Gross2006,veitch2012negative,delfosse2017equivalence}, but, interestingly, does not require the full Wigner function, only its definition at a single point in phase-space.
While a discrete Wigner function can be defined in even dimensions based on generalized Pauli observables \cite{Gibbons_2004}, it is easy to show that a satisfiability gap can be demonstrated using generalized Pauli observables in every even dimension (see \Cref{tab:table}). 

For a LCS over $\bb{Z}_d$ given by the equation $Mx = b \mod d$, checking whether a classical solution exists is easy, for example by Gaussian elimination if $d$ is prime, or other efficient methods if $d$ is not prime \cite{ching1977linear,storjohann1998fast}. On the other hand, checking whether a quantum solution exists is a difficult task. One characterization is given in \cite{cleve2017perfect} in terms of the so-called solution group.
The solution group associated with $Mx = b$ is a finitely presented group which encodes the algebraic relations satisfied by every quantum solution to $Mx=b$. 
It is generated by symbols $\{g_i\}$, which can be thought of as a non-commutative relaxation of the variables $\{x_i\}$, in addition to a special generator $J$. 
The characterization in \cite{cleve2017perfect} is based on properties of the element $J$. 
Namely, for a LCS over $\bb{Z}_d$ with $d$ prime, a quantum solution exists if and only if $J$ is not equivalent to the identity element in the group.
Checking whether an element in a finitely presented group is equivalent to the identity is an instance of the word problem, which is undecidable in general.
Nevertheless, for small enough instances or in the case of highly structured groups one can use computer software, such as GAP \cite{GAP4}, to find a so-called confluent rewriting system, thus solving the word problem for the group.

In this paper, we note a simple extension of the characterization in \cite{cleve2017perfect} which covers arbitrary $d$. Specifically, a quantum solution exists for a LCS over $\bb{Z}_d$ if and only if $J$ has order $d$.
When $d$ is prime, this characterization reduces to the one in \cite{cleve2017perfect}.

Furthermore, we note a key property of some LCS, which we call phase-commutation. Roughly speaking, phase-commutation is an abstraction of matrix anti-commutation relations of the form $AB = - BA$.
A LCS $Mx  = b$ exhibits phase-commutation if its non-commutative relaxation (as defined by the solution group) contains two variables, $x_1$ and $x_2$, which commute up to a constant; 
\begin{align}\label{eq: intro}
x_1 + x_2   = c + x_2 + x_1, \quad c\neq 0.
\end{align}
Importantly, such relations arise in the LCSs of Peres and Mermin. 
In the current paper, we show that, if $d$ is odd, any relation of the form in \cref{eq: intro}, or of the more general form in \cref{eq: generalized pc}, implies that quantum solutions do not exist in any (finite- or infinite-dimensional) Hilbert space.
As a consequence, this implies that all natural generalizations of the Peres-Mermin magic square and pentagram to odd prime $d$ do not exhibit a satisfiability gap. 

This is in contrast to the case of even $d$, where the Peres-Mermin magic square and pentagram exhibit satisfiability gaps that rely on using Pauli operators to represent \cref{eq: intro}.
In fact, a much stronger statement was proven for $d=2$ in Ref.~\cite{slofstra2019tsirelson}, namely that any group generated by involutions can be embedded (in a strong sense) in the solution group of some binary LCS which admits a quantum solution. 
A significant consequence of this embedding for $d = 2$ is the separation, proven in the latter work, between commuting and tensor-product quantum nonlocal strategies.
Our result shows that a similar statement cannot hold for odd $d$, as the group defined by the phase-commutation relation in \cref{eq: intro} cannot be embedded in a LCS that has a quantum solution.

The idea that phase-commutation can rule out a satisfiability gap appears in (the newer version of) Ref. \cite{coladangelo2017robust} for particular instances of the magic square and pentagram LCSs, and a similar idea appears more recently in \cite{okay2019homotopical}, as part of a homotopical treatment of quantum contextuality. In comparison, our exposition is more general (it applies to all LCSs, and we use it to cover all natural generalizations of the square and pentagram for odd $d$), and perhaps more accessible to the non-expert reader (it does not require familiarity with topological ideas like cell-complexes and group pictures). 

\section{Preliminaries}
A linear constraint system (LCS) over $\bb{Z}_d$ is a pair $(M,b)$, where $M \in \bb{M}_{m\times n}(\bb{Z}_d)$ and $b \in \mathbb{Z}_d^m$.
A classical solution is a vector $a \in \mathbb{Z}_d^n$ such that $Ma=b$ modulo $d$.
\Cref{fig:Mermin}~(a) is a diagrammatic depiction of a LCS that is classically unsatisfiable for all $d\geq 2$, since the sum of the variables along the rows is 0, while the sum of the variables along the columns is 1, which is a contradiction.

A quantum solution $A \equiv (A_1, \dots, A_n)$ is a collection of operators, acting on some Hilbert space, satisfying:
\begin{enumerate}
\item $A_i^d = I$ for all $i \in \{1, \dots, n\}$;
\item $[A_j,A_k] = 0$ whenever there exists a row $i$ such that $M_{ij}\neq 0$ and $M_{ik} \neq 0$; and
\item $\prod_{j} A_{j}^{M_{ij}} = \omega^{b_i}$ for all $i \in \{1, \dots, m\}$.
\end{enumerate}
\begin{figure}[t]
\centering
\providecommand{\outerRadius}{1.1}
\providecommand{\innerRadius}{\outerRadius / 2.61803}
\begin{tikzpicture}[scale=2]
    \node at (0, 0) [rectangle] (v00) {$x_{7}$};
    \node at (1, 0) [rectangle] (v10) {$x_{8}$};
    \node at (2, 0) [rectangle] (v20) {$x_{9}$};
    \node at (0, 1) [rectangle] (v01) {$x_{4}$};
    \node at (1, 1) [rectangle] (v11) {$x_{5}$};
    \node at (2, 1) [rectangle] (v21) {$x_{6}$};
    \node at (0, 2) [rectangle] (v02) {$x_{1}$};
    \node at (1, 2) [rectangle] (v12) {$x_{2}$};
    \node at (2, 2) [rectangle] (v22) {$x_{3}$};
    \draw (v00) -- (v10);
    \draw (v10) -- (v20);
    \draw (v01) -- (v11);
    \draw (v11) -- (v21);
    \draw (v02) -- (v12);
    \draw (v12) -- (v22);
    \draw (v00) -- (v01);
    \draw (v10) -- (v11);
    \draw[dashed] (v20) -- (v21);
    \draw (v02) -- (v01);
    \draw (v12) -- (v11);
    \draw[dashed] (v22) -- (v21);
\end{tikzpicture}
\hspace{5mm}
\begin{tikzpicture}[scale=2]
\node at (90:\outerRadius) [rectangle] (v0) {$x_8$};
\node at (90 + 72*1:\outerRadius) [rectangle] (v1) {$x_1$};
\node at (90 + 72*2:\outerRadius) [rectangle] (v2) {$x_9$};
\node at (90 + 72*3:\outerRadius) [rectangle] (v3) {$x_6$};
\node at (90 + 72*4:\outerRadius) [rectangle] (v4) {$x_4$};
\node at (-90:\innerRadius) [rectangle] (v5) {$x_7$};
\node at (-90 + 72*1:\innerRadius) [rectangle] (v6) {$x_{10}$};
\node at (-90 + 72*2:\innerRadius) [rectangle] (v7) {$x_3$};
\node at (-90 + 72*3:\innerRadius) [rectangle] (v8) {$x_2$};
\node at (-90 + 72*4:\innerRadius) [rectangle] (v9) {$x_{5}$};
\draw (v0) -- (v8);
\draw (v8) -- (v9);
\draw (v9) -- (v2);
\draw (v0) -- (v7);
\draw (v7) -- (v6);
\draw (v6) -- (v3);
\draw (v1) -- (v9);
\draw (v9) -- (v5);
\draw (v5) -- (v3);
\draw (v4) -- (v6);
\draw (v6) -- (v5);
\draw (v5) -- (v2);
\draw[dashed] (v1) -- (v8);
\draw[dashed] (v8) -- (v7);
\draw[dashed] (v7) -- (v4);
\end{tikzpicture}
    \caption{Linear constraint systems for the magic square (left) and pentagram (right)~\cite{Mermin1993}.
    Each variable can take values modulo $d$.
    The linear constraints are that the values of the variables along a solid (dashed) line add to zero (one) modulo $d$.
    Note that whether the games are ``magic'' depends upon the constraints, e.g., if all lines are solid then both games are satisfiable.}
    \label{fig:Mermin}
\end{figure}
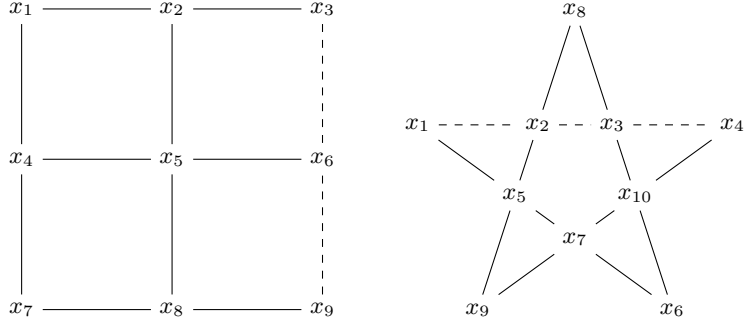
By condition (1), each operator $A_i$ is unitary and the eigenvalues are $d$th roots of unity.
By condition (2), whenever $x_j$ and $x_k$ appear in one of the constraints, their operator counterparts $A_j$ and $A_k$ commute.
In particular, all operators corresponding to a constraint pairwise commute, and therefore can be simultaneously measured.
By condition (3), the expectation values of the operators satisfy the specified constraints.
Here and below, we set $\omega \equiv \exp 2\pi i /d$, where $d$ is to be understood from context.
Quantum solutions for the magic square and pentagram are shown in \Cref{fig:SAT}.

For a set $S=\{ s, t, u, \dots \}$, the \textit{free group} $\mc{F}(S)$ generated by $S$ consists of all words in the letters $S \cup S'$, where $S'=\{ s^{-1}, t^{-1}, u^{-1}, \dots \}$. Two words correspond to the same element in $\mc{F}(S)$ if one of them can be reduced to the other using contractions of the form $e =s^{-1}s =ss^{-1}$, where $e$ is the identity element.
Let $G$ be a group. 
The \textit{normal closure} of a subset $R \subseteq G$ is the normal subgroup $N(R)$ generated by the set $\{ grg^{-1} : g \in G, r \in R\}$. 
Let $R = \{r_1, \dots, r_k \}$ be a set of words in $\mc{F}(S)$. The \textit{finitely presented group} with generators $S$ and relations $\{ r_1 =e , \dots, r_k =e \}$ is the quotient $\mc{F}(S)/N(R)$. For further reading related to these definitions, see any standard reference on abstract algebra, for instance \cite{grillet2007abstract}.

To every linear constraint system one can associate a \textit{solution group}, which is a finitely presented group with relations that encode the defining properties of a quantum solution. Namely,
\begin{definition}
For a LCS $(M,b)$, the \textit{solution group} $\mc{G}(M,b)$ is the finitely presented group generated by the symbols $\{ J, g_1, \dots, g_n \}$ and the relations
\begin{enumerate}
\item $g_i^d = e$, $J^d = e$,
\item $J g_i J^{-1} g_i^{-1} = e$,
\item $g_j g_k g_j^{-1} g_k^{-1} = e$ whenever there exists a row $i$ such that $M_{ij}\neq 0$ and $M_{ik} \neq 0$,
\item $J^{-b_i} \prod_{j} g_{j}^{M_{ij}} = e$ for all $i \in \{1, \dots, m\}$.
\end{enumerate}
\end{definition}

\begin{proposition}\label{reps of sg}
Let $(M,b)$ be a LCS. Suppose that a quantum solution exists in dimension $D \in \mathbb{N} \cup \{ \infty\}$. Then $\mc{G}(M,b)$ admits a $D$-dimensional representation $\phi$ in which $\phi(J) = \omega I$. 
\end{proposition}
\begin{proof}
Given a quantum solution $\{A_1, \dots, A_n \}$ in dimension $D$, it is straight-forward to check that the homomorphic extension of the map $\phi(g_i) = A_i$, $\phi(J) = \omega I$, is a $D$-dimensional representation of $\mc{G}(M,b)$. 
\end{proof}

\begin{figure}
\centering
\providecommand{\outerRadius}{1.1}
\providecommand{\innerRadius}{\outerRadius / 2.61803}
\begin{tikzpicture}[scale=2]
    \node at (0, 0) [rectangle] (v00) {$IX$};
    \node at (1, 0) [rectangle] (v10) {$XI$};
    \node at (2, 0) [rectangle] (v20) {$XX$};
    \node at (0, 1) [rectangle] (v01) {$ZI$};
    \node at (1, 1) [rectangle] (v11) {$IZ$};
    \node at (2, 1) [rectangle] (v21) {$ZZ$};
    \node at (0, 2) [rectangle] (v02) {$ZX$};
    \node at (1, 2) [rectangle] (v12) {$XZ$};
    \node at (2, 2) [rectangle] (v22) {$YY$};
    \draw (v00) -- (v10);
    \draw (v10) -- (v20);
    \draw (v01) -- (v11);
    \draw (v11) -- (v21);
    \draw (v02) -- (v12);
    \draw (v12) -- (v22);
    \draw (v00) -- (v01);
    \draw (v10) -- (v11);
    \draw[dashed] (v20) -- (v21);
    \draw (v02) -- (v01);
    \draw (v12) -- (v11);
    \draw[dashed] (v22) -- (v21);
\end{tikzpicture}
\hspace{0mm}
\begin{tikzpicture}[scale=2]
\node at (90:\outerRadius) [rectangle] (v0) {$YII$};
\node at (90 + 72*1:\outerRadius) [rectangle] (v1) {$XXX$};
\node at (90 + 72*2:\outerRadius) [rectangle] (v2) {$IYI$};
\node at (90 + 72*3:\outerRadius) [rectangle] (v3) {$IXI$};
\node at (90 + 72*4:\outerRadius) [rectangle] (v4) {$XYY$};
\node at (-90:\innerRadius) [rectangle] (v5) {$XII$};
\node at (-90 + 72*1:\innerRadius) [rectangle] (v6) {$IIY$};
\node at (-90 + 72*2:\innerRadius) [rectangle] (v7) {$YXY$};
\node at (-90 + 72*3:\innerRadius) [rectangle] (v8) {$YYX$};
\node at (-90 + 72*4:\innerRadius) [rectangle] (v9) {$IIX$};
\draw (v0) -- (v8);
\draw (v8) -- (v9);
\draw (v9) -- (v2);
\draw (v0) -- (v7);
\draw (v7) -- (v6);
\draw (v6) -- (v3);
\draw (v1) -- (v9);
\draw (v9) -- (v5);
\draw (v5) -- (v3);
\draw (v4) -- (v6);
\draw (v6) -- (v5);
\draw (v5) -- (v2);
\draw[dashed] (v1) -- (v8);
\draw[dashed] (v8) -- (v7);
\draw[dashed] (v7) -- (v4);
\end{tikzpicture}
\caption{Quantum solutions for the magic square (left) and pentagram (right) for $d=2$~\cite{Mermin1993}.
    Each vertex is a unitary operator whose eigenvalues are $d$th roots of unity.
    The operators along a solid (dashed) line multiply to $I$ ($-I$) modulo $d$.}
    \label{fig:SAT}
\end{figure}
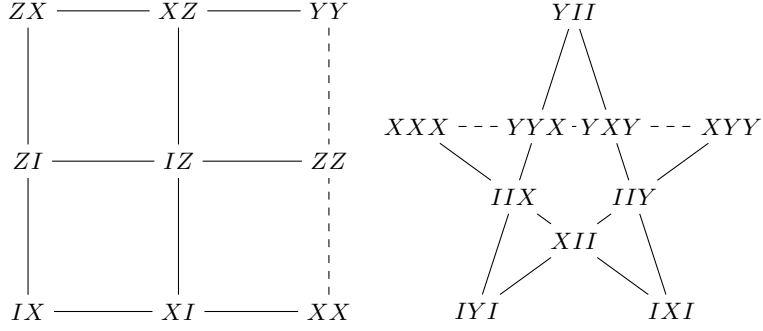
The single-qubit Pauli operators can be generalized to arbitrary dimensions $D$ via the generalized \emph{shift} and \emph{boost} operators
\begin{align}\label{eq:defineXZ}
	X &= \sum_{j\in\mathbb{Z}_D} \ketbra{j + 1}{j} \notag\\
	Z &= \sum_{j\in\mathbb{Z}_D} \omega^j \ketbra{j} \,.
\end{align}
The generalized shift and boost operators satisfy the commutation relation
\begin{align}\label{eq:XZcommutation}
	Z^p X^q = \omega^{-pq} X^q Z^p	\,.
\end{align}
A generalized Pauli operator for a single qu$D$it is any operator $A\propto Z^p X^q$ for some $p,q\in\mathbb{Z}_D$.
For $n$ qu$D$its, a generalized Pauli operator is defined as a tensor product of single qu$D$it generalized Paulis.
From \cref{eq:XZcommutation}, any two generalized Pauli operators commute up to a $D$th root of unity.

\section{Quantum solutions using generalized Pauli operators}

\begin{theorem}\label{thm:pauli_assignment}
Let $D$ be odd. Then the LCS $(M,b)$ has a quantum solution using $n$-qu$D$it generalized Pauli operators if and only if it is classically satisfiable.
\end{theorem}
\begin{proof}
The first direction is easy: suppose that $Mx = b$ has a classical solution $x \rightarrow  a \equiv (a_1, \dots, a_n)$. Then $A_j = \omega^{a_j} I$ defines a quantum solution in which every operator is a Pauli operator. Conversely, let us assume that the LCS admits a quantum solution $\{A_j\}$ consisting of generalized $n$-qu$D$it Pauli operators. We will construct a classical solution by consistently assigning a root of unity $\omega^{v(\sigma)}$ to each generalized $n$-qu$D$it Pauli operator $\sigma$. The map $v$ will be defined in terms of the parity operator $\Pi_D \equiv \sum_{j=0}^{D-1} \ketbra{-j}{j}$. Namely, we set
\begin{align}\label{map}
\omega^{v(\sigma)} = \tr \Pi_D\tn{n} \sigma.
\end{align}
For a single-qu$D$it Pauli $\sigma = \omega^k Z^p X^q$, \cref{map} gives $v(\sigma) = k+2^{-1} pq $, where the inverse is modulo $D$. 
It can be readily verified that the following identities hold for any $n$-qu$D$it generalized Paulis $\sigma,\tau$
\begin{enumerate}
\item $v(\omega^{j}I) =j$,
\item $v(\sigma \otimes \tau) = v(\sigma) + v(\tau) \mod D$ ,
\item $v(\sigma^{j}) = j v(\sigma) \mod D$,
\end{enumerate}
To prove that the map $A_j \rightarrow v(A_j)$ gives a classical solution to the LCS in question, it is enough to show that, if $\sigma$ and $\tau$ are commuting $n$-qu$D$it generalized Pauli operators, then $v(\sigma\tau) = v(\sigma)+ v(\tau)$.
Indeed if that were true then using identities $1$ \& $3$ we get, for every row $i$ of $M$, 
\begin{align}
\sum_{j =1}^{n} M_{i,j} v\left(A_j^{}\right) = \sum_{j =1}^{n} v\left(A_j^{M_{i,j}}\right) = v\left( \prod_{j=1}^n A_j^{M_{i,j}} \right) = v\left( \omega^{b_i} I\right) = b_i.
\end{align}
Let $\sigma \propto Z[\vec{p}] X[\vec{q}]$ and $\tau \propto Z[\vec{r}] X[\vec{s}]$ be tensor products of generalized Pauli operates.
Here we are denoting $A[\vec{a}] = \bigotimes_{k=1}^n A^{a_k}$ for an $n$-element vector $\vec{a} = (a_1,\ldots, a_n)$, and, for simplicity, ignoring the phase factors in the definition of generalized Paulis, which does not affect the analysis.
Note that 
\begin{align}\label{eq:product_value}
	v(Z^p X^q) + v(Z^r X^s) = 2^{-1}(pq+rs)	\ \forall\, p,q,r,s.
\end{align}
On the other hand, by \cref{eq:XZcommutation},
\begin{align}\label{eq:value_product}
	v(Z^p X^q Z^r X^s) = v(\omega^{-qr} Z^{p+r}X^{q+s}) = - qr + 2^{-1}(p+r)(q+s)	\,.
\end{align}
Comparing Eq.~\eqref{eq:product_value} and \eqref{eq:value_product} and using identity $2$,
\begin{align}
	v(Z[\vec{p}]X[\vec{q}] Z[\vec{r}]X[\vec{s}]) = v( Z[\vec{p}]X[\vec{q}])+ v( Z[\vec{r}]X[\vec{s}]) \Longleftrightarrow \vec{p}\cdot\vec{s} - \vec{q}\cdot\vec{r} = 0 \ ({\rm mod}\ d)	\,.
\end{align}
The condition $\vec{p}\cdot\vec{s} - \vec{q}\cdot\vec{r} = 0$ is precisely the requirement that $	Z[\vec{p}]X[\vec{q}]$ and $Z[\vec{r}]X[\vec{s}]$ commute.  Indeed, by \cref{eq:XZcommutation}
\begin{align}
	Z[\vec{p}]X[\vec{q}] Z[\vec{r}]X[\vec{s}] = \omega^{\vec{q}\cdot \vec{r}}Z[\vec{p}] Z[\vec{r}]X[\vec{s}]X[\vec{q}] = \omega^{\vec{q}\cdot \vec{r} - \vec{p}\cdot \vec{s}} Z[\vec{r}]X[\vec{s}] Z[\vec{p}]X[\vec{q}]\,.
\end{align}
Therefore $v(\sigma\tau) = v(\sigma) + v(\tau)$ for any commuting $\sigma$ and $\tau$, finishing the proof.
\end{proof}
The above proof fails for even $D$.  
The proof hinges on the ability to consistently assign a root of unity to every generalized Pauli operator. For $\sigma = \omega^k Z^p X^q$, \cref{map} reads
\begin{align}
\omega^{v(\sigma)} = \tr \Pi_D \sigma = \sum_j \bra{j}\sigma\ket{-j} = \sum_{j: 2j = q} \omega^{pj +k}.
\end{align}
When $D$ is even, the sum evaluates to $0$ whenever $q$ is odd, making it impossible to define the map in this way.
One may wonder whether the map $v$ can be defined in a different way for even $D$. This is also impossible, as shown by the family of quantum solutions in \Cref{tab:table}.
Note that the parity operator $\Pi_D$ in the proof of \Cref{thm:pauli_assignment} is one of the phase point operators for the discrete Wigner function in odd dimensions~\cite{Gross2006}. 
While it is possible to define a discrete Wigner function in even dimensions based on generalized Pauli operators \cite{Gibbons_2004}, important properties, such as non-negativity of stabilizer states and Clifford-covariance, are lost in that case \cite{Zhu_2016}.
\begin{table}[]
\begin{tabular}{|c|c|c|}
\hline
 $Z^{t}\otimes X^{t} $ & $X^{t}\otimes Z^{t}$  & $X^{t}Z^{t} \otimes Z^{t}X^{t}$  \\ \hline
 $Z^{t}\otimes I$ & $I\otimes Z^{t}$  & $Z^{t}\otimes  Z^{t}$ \\ \hline
 $I\otimes X^{t}$ & $X^{t}\otimes I$  & $X^{t}\otimes X^{t}$  \\ \hline
\end{tabular}
\quad\quad
\begin{tabular}{|c|c|c|}
\hline
 $Z^{t+1}\otimes X^{t+1} $ & $X^{t}\otimes Z^{t}$  & $X^{t}Z^{t-1} \otimes Z^{t}X^{t-1}$  \\ \hline
 $Z^{t-1}\otimes I$ & $I\otimes Z^{t}$  & $Z^{t+1}\otimes  Z^{t}$ \\ \hline
 $I\otimes X^{t-1}$ & $X^{t}\otimes I$  & $X^{t}\otimes X^{t+1}$  \\ \hline
\end{tabular}
\caption{Quantum solutions to a generalized magic square LCS modulo $D=2t$, for odd (even) $t$ on the left (right). In both cases, the generalized two-qu$D$it Pauli operators along every row and column mutually commute, their product along the third column is $\omega^t I = -I$, and their product along the rows and remaining columns is $I$. This shows that a satisfiability gap exists using n-qu$D$it generalized Pauli operators for every even dimension $D$. }
\label{tab:table}
\end{table}
\section{Phase-commutation and existence of a quantum solution}
\Cref{thm:pauli_assignment} asserts that a quantum/classical satisfiability gap cannot be illustrated using generalized Pauli operators in odd dimension.
We now show that the canonical commutation relations of these operators cannot be embedded in the solution group of a LCS that has a quantum solution.
Some of this discussion is inspired by (the corrected version of) Ref.~\cite{coladangelo2017robust}, but it applies more generally than what is considered in that work.

A key characterization of quantum satisfiability when $d$ is prime is the following:
\begin{theorem}{(Cleve et al, \cite{cleve2017perfect}).}\label{thm: cleve}
Let $(M,b)$ be a LCS over $\mathbb{Z}_d$ where $d$ is prime. Then $(M,b)$ admits a quantum solution if and only if the solution group $\mc{G}(M,b)$ has the property that $J \neq e$.
\end{theorem}
The proof of \Cref{thm: cleve} relies on the fact that, when $J \neq e$, the group algebra of $\mc{G}(M,b)$ contains a subspace on which it is easy to construct a quantum solution.
It turns out that the proof works for non-prime $d$ as well, provided we make an adjustment to the statement of the theorem.
\begin{theorem}\label{thm: char1}
Let $(M,b)$ be a LCS over $\mathbb{Z}_d$ with $d > 1$. Then $(M,b)$ admits a quantum solution if and only if the order of $J$ is $d$.
\end{theorem}
The proof of \Cref{thm: char1} does not differ much from that of \Cref{thm: cleve}, which can be found in \cite{cleve2017perfect}. We provide a somewhat more concise version of it in \Cref{app: proof}.
Note that \Cref{thm: cleve} is a special case of \Cref{thm: char1} since, for prime $d$, cyclicity of the powers  of $J$ imply that $J\neq e$ if and only the order of $J$ is $d$.\\

Generalized Pauli operators obey commutation relations of the form $AB = \omega^c BA$, and are in a sense uniquely determined by these relations. 
More specifcally, consider the finite group $\mc{P}_d$ with generators $\{\mc{X},\mc{Z},\eta\}$ and relations
\begin{align}\label{eq: CCR}
\mc{X}\mc{Z}= \eta \mc{Z}\mc{X},\quad \mc{X}^d = \mc{Z}^d = \eta^{d} = e, \quad \eta \mc{X} = \mc{X} \eta, \quad \eta \mc{Z} = \mc{Z} \eta.
\end{align}
By Schur's lemma and the fact that $\eta^d = e$, every irreducible representation of $\mc{P}_d$ assigns a $d$th root of unity $\omega^{j} I$ to $\eta$.
Furthermore, it was shown by Weyl that up to unitary equivalence every irreducible representation of $\mc{P}_d$ in which $j \neq 0$ maps $\mc{X}$ and $\mc{Z}$ to (phase-multiples of) the clock and shift operators in some dimension $d'$ that divides $d$ \cite{weyl1950theory}.

We consider a particular notion of an embedding of \cref{eq: CCR} into the solution group of a LCS.
Namely, suppose that the solution group contains a relation of the form 
\begin{align}\label{eq: pc}
g^a h^b  = J^c h^b g^a 
\end{align}
for a pair of generators $g$ and $h$ and some nonzero $a,b,c \in \bb{Z}_d$.
We refer to such a relation as a phase-commutation relation in the solution group.
A phase-commutation relation implies that the LCS is classically unsatisfiable;  no integers $x$ and $y$ can satisfy the equation $a x + b y = c + b y + a x   \mod d$ for $c\neq 0 \mod d$.

This notion of an embedding can be motivated by considering the situation for $d = 2$.
Given a group $K$, a collection of elements $k_1, \dots, k_m \in K$ and a non-identity central element $\eta$, such that $k_i^2 = \eta^2 = e$, it was shown in \cite{slofstra2019tsirelson} that there exists a LCS $Mx = b$ and an embedding $\psi: K \rightarrow \mc{G}(M,b)$, such that $\psi(k_i) = g_i$ and $\psi(\eta) = J$.
As $\psi$ is injective, this implies that $J \neq e$, and so every such $K$ embeds in the solution group of some LCS which has a (possibly infinite-dimensional) quantum solution.

In contrast, we show next that, for odd $d$, a phase-commutation relation implies that quantum solutions to the LCS in question do not exist in any dimension, which rules out a satisfiability gap for any LCS in which \cref{eq: CCR} embeds in the solution group.

\begin{theorem}\label{thm: phase commutation}
Let $d$ be odd. Suppose that a LCS over $\bb{Z}_d$ has a phase-commutation relation. Then it has no quantum solutions.
\end{theorem}
\Cref{thm: phase commutation} can be proven diagrammatically using symmetries of the so-called group-picture associated with solution groups \cite{slofstra2019tsirelson,coladangelo2017robust}. 
We avoid group-pictures for the sake of simplicity, 
and present a more direct proof instead.
\begin{proof} Let $S \equiv \{ J, g_1, \dots, g_n\}$ denote the generators of the solution group $\mc{G}$, and $S' \equiv \{ J^{-1}, g_1^{-1}, \dots, g_n^{-1}\}$. 
Let $R \subset \mc{F}(S)$ be the set of relations defining $\mc{G}$, so that $\mc{G}= \mc{F}(S)/N(R)$. 
Let $g \in S$ and $h \in S$ be a phase-commuting pair, so that, in $\mc{G}$,
\begin{align}\label{eq: pc def}
g^a h^b   = J^c h^b g^a, \quad c\neq 0.
\end{align}
Let $w$ be the word $J^{-c} g^ah^bg^{-a} h^{-b}$.
Since $w = e$ in $\mc{G}$, it holds that, in $\mc{F}(S)$, $w$ is equivalent to an element of the normal closure $N(R)$. This implies that there exist $u_1, \dots, u_{\ell} \in \mc{F}(S)$ and $q_1, \dots, q_\ell \in R$, such that the word
\begin{align}\label{word1}
W \equiv (u_1 q_1 u_1^{-1}) (u_2 q_2 u_2^{-1}) \dots (u_\ell q_\ell u_\ell^{-1}),
\end{align}
is equivalent to $w$ in $\mc{F}(S)$, that is, $W$ can be reduced to $w$ using only generator contractions of the form $g_i g_i^{-1} = e$.
Let $\tilde{W}$ be the ``reflected'' version of $W$
\begin{align}
\tilde{W} \equiv  (\tilde{u}_\ell^{-1} \tilde{q}_\ell \tilde{u}_\ell) \dots  (\tilde{u}_2^{-1} \tilde{q}_2 \tilde{u}_2)  (\tilde{u}_1^{-1} \tilde{q}_1 \tilde{u}_1),
\end{align}
where $\tilde{x}$ is obtained from $x$ by reversing the order of multiplication of the letters making up $x$. For example, if $x = g_1 g_2 g_3$ then $\tilde{x} = g_3 g_2 g_1$. 

Recall, from the definition of solution groups, that each of the relations in $R$ (and therefore each of the $q_i$) consists of mutually commuting letters from $S \cup S'$. It therefore holds that $\tilde{q}_i = q_i = e$ in $\mc{G}$. This implies that $\tilde{W} = e$ in $\mc{G}$. 
Furthermore, we have $\tilde{W}=\tilde{w}$ in $\mc{G}$, since the same sequence of generator contractions that reduces $W$ to $w$ can be applied to reduce $\tilde{W}$ to $\tilde{w}$, so $\tilde{W}$ is equivalent to $\tilde{w}$ in $\mc{F}(S)$, and therefore also in $\mc{G}$.

Thus we have $\tilde{w} = e$ in $\mc{G}$, and therefore $h^b g^a = J^{c} g^a h^b$ in $\mc{G}$. Combining this with $g^a h^b = J^c h^b g^a$ implies that $J^{2c} =e$ in $\mc{G}$. Since $d$ is odd, we have $2c \neq 0 \mod d$. Therefore, by \Cref{thm: char1}, there are no quantum solutions.
\end{proof}

The statement and proof of \Cref{thm: phase commutation} can be generalized by considering relations of the form
\begin{align}\label{eq: generalized pc}
g_1^{a_1} \dots g_k^{a_k} = J^c g_k^{a_k} \dots g_1^{a_1}.
\end{align}
Following the same steps in the proof of \Cref{thm: phase commutation}, we can deduce that $J^{2c} = e$, which, for $c \neq 0 \mod d$, implies that the order of $J$ is less than $d$, and hence that there are no quantum solutions. 
\section{Examples}\label{sec: examples}
\subsection{Magic square modulo prime $d>2$}
Let $d>2$ be prime. Here we consider a broad class of generalizations of the magic square LCS. Namely, we consider any LCS، over $\bb{Z}_d$ of the form in \Cref{fig:Mermin}~(a) with arbitrary coefficients, i.e. any system of the form
\begin{align} \label{ms dlc}
&a_1 x_1 + a_2 x_2 + a_3 x_3 = b_1, \nonumber\\
&a_4 x_4 +  a_5 x_5 + a_6 x_6 = b_2, \nonumber\\
& a_7 x_7 + a_8 x_8 + a_9 x_9 = b_3, \nonumber\\
&a_{1}' x_1 + a_{4}' x_4 + a_{7}' x_7 = b_4, \nonumber\\
&a_{2}' x_2 + a_{5}' x_5 + a_{8}' x_8 = b_5, \nonumber\\
&a_{3}'  x_3  + a_{6}' x_6 + a_{9}' x_9 = b_6,
\end{align}
for some $b \in \bb{Z}_d^6$, and some set of coefficients $a_i, a_i' \in \bb{Z}_d - \{ 0\}$.

We can simplify by making some reductions.
First, we assume that the coefficients in the first three constraints are all equal to one, if necessary by relabeling $x_i \rightarrow a_i x_i$.
Second, as $d$ is prime, every nonzero element of $\bb{Z}_d$ has a multiplicative inverse, so we can divide the last three constraints respectively by the coefficients of $x_1$, $x_2$, and $x_3$. These two steps yield a LCS of the form
\begin{align} \label{ms dlc reduced}
 x_1 + x_2 + x_3 &= b_1 \nonumber\\
 x_4 + x_5 + x_6 &= b_2 \nonumber\\
 x_7 + x_8 + x_9 &= b_3 \nonumber\\
 x_1 +{\gamma_4} x_4 +{\gamma_7} x_7 &= b_4 \nonumber\\
 x_2 + {\gamma_5} x_5 + {\gamma_8} x_8 &= b_5 \nonumber\\
 x_3 + {\gamma_6} x_6 + {\gamma_9} x_9 &= b_6,
\end{align}
for some $\gamma_i \in \bb{Z}_d - \{ 0\}$, and (generally different) $b \in \bb{Z}^6_d$.
Note that these reductions preserve classical and quantum satisfiability. Denote the above reduced LCS by $M x = b$, where 
\begin{align}
M = \left(
\begin{array}{ccccccccc}
 1 & 1 & 1 & 0 & 0 & 0 & 0 & 0 & 0 \\
 0 & 0 & 0 & 1 & 1 & 1 & 0 & 0 & 0 \\
 0 & 0 & 0 & 0 & 0 & 0 & 1 & 1 & 1 \\
 1 & 0 & 0 &\gamma_4 & 0 & 0 & \gamma_7 & 0 & 0 \\
 0 & 1 & 0 & 0 & \gamma_5& 0 & 0 &\gamma_8& 0 \\
 0 & 0 & 1 & 0 & 0 & \gamma_6& 0 & 0 & \gamma_9 \\
\end{array}
\right).
\end{align}
It is straightforward to check that
\begin{align}\label{ms rank}
\text{rank}(M) = 
\begin{cases}
5 & \text{if } \gamma_4 = \gamma_5 = \gamma_6 \text{ and } \gamma_7 = \gamma_8 = \gamma_9,\\
6 & \text{otherwise},
\end{cases}
\end{align}
where the rank is computed modulo $d$ (which is well-defined as $d$ is prime).
Thus the system $Mx =b$ is classically unsatisfiable only if $\gamma_4 = \gamma_5 = \gamma_6$ and $\gamma_7 = \gamma_8 = \gamma_9$.
Essentially this means that we only need to check $(d-1)^2$ cases, corresponding to the values of the $\gamma_i \in \bb{Z}_d - \{ 0\}$ that could make the LCS classically unsatisfiable.
As we will see next, in each of these cases a phase-commutation relation prohibits the existence of a satisfiability gap.
\begin{theorem}\label{ms thm}
Let $d > 2$ be a prime integer. A LCS of the form in \cref{ms dlc} over $\bb{Z}_d$ has a quantum solution if and only if it is classically satisfiable.
\end{theorem}
\begin{proof}
If a classical solution exists then a quantum solution exists.
Conversely, assume that a quantum solution exists, and denote the associated reduced LCS in \cref{ms dlc reduced} by $Mx=b$. Assume that $\gamma_4=\gamma_5=\gamma_6\equiv \gamma$ and $\gamma_7=\gamma_8=\gamma_9 \equiv \delta$, otherwise there is a classical solution and we are done. 
In the solution group $\mc{G}(M,b)$, we have
\begin{align} \label{ms dlc os}
g_1  g_2  g_3 &= J^{b_1} \quad
g_4  g_5  g_6 = J^{b_2} \quad
g_7  g_8  g_9 = J^{b_3} \nonumber\\
g_1  g_4^{\gamma}  g_7^{\delta} &= J^{b_4} \quad
g_2  g_5^{\gamma}  g_8^{\delta} = J^{b_5} \quad
g_3  g_6^{\gamma}  g_9^{\delta} = J^{b_6}.
\end{align}
We note the group identity
\begin{align}
g_1  g_4^{\gamma}g_2  g_5^{\gamma}
&= J^{b_4 + b_5} g_7^{-\delta} g_8^{-\delta} \nonumber\\
&= J^{-\delta b_3 + b_4 + b_5} g_9^{\delta} \nonumber\\
&= J^{-\delta b_3 + b_4 + b_5+b_6} g_3^{-1} g_6^{- \gamma } \nonumber\\
&= J^{b_1-\gamma b_2-\delta b_3 + b_4 + b_5+b_6} g_1 g_2 g_4^{\gamma } g_5^{\gamma}.
\end{align}
Multiplying both sides on the left by $g_1^{-1}$ and on the right by $g_5^{-\gamma}$ reveals the phase-commutation
\begin{align}\label{ms pc}
g_2 g_4^{\gamma} = J^{-b_1-\gamma b_2-\delta b_3 + b_4 + b_5+b_6} g_4^{\gamma} g_2. 
\end{align}
Since $d$ is odd, by \Cref{thm: phase commutation} and \cref{ms pc}, a necessary condition for the existence of a quantum solution is that $-b_1-\gamma b_2-\delta b_3 + b_4 + b_5+b_6 = 0 \mod d$.
This implies that the LCS is classically satisfiable: set $x_b = (-\gamma b_2-\delta b_3+b_4,b_5,b_1+\gamma b_2+\delta b_3-b_4-b_5,b_2,0,0,b_3,0,0)$, then $M x_b = (b_1,b_2,b_3,b_4,b_5,b_1+\gamma b_2+\delta b_3 - b_4 - b_5) = b$. This shows that the original LCS in \cref{ms dlc} has a classical solution.
\end{proof}
\subsection{Magic pentagram modulo prime $d>2$}
Similar to the magic square LCS, we consider any LCS، over $\bb{Z}_d$ of the pentagram form in \Cref{fig:Mermin}~(b) with arbitrary coefficients, i.e. any system of the form
\begin{align}\label{mp dlc}
a_1 x_1  +  a_2   x_2 +  a_3  x_3  +  a_4   x_4 &= b_1	\notag\\
a_1' x_1 +  a_5  x_5   +  a_6  x_6  +  a_7   x_7 &= b_2	\notag\\
a_2'  x_2  +  a_5'   x_5  +  a_8   x_8  +  a_9   x_9 &= b_3	\notag\\
a_3' x_3  +  a_6'   x_6  +  a_8'   x_8  +   a_{10}   x_{10} &= b_4	\notag\\
a_4' x_4  +  a_7'   x_7  +  a_9'   x_9  +  a_{10}'   x_{10} &= b_5.	
\end{align}  
We can perform similar reductions to the magic square case.
First, we set all coefficients equal to 1 in the first constraint by relabeling $x_i\rightarrow  a_i x_i$ if necessary.
Next, we divide the remaining four constraints by the coefficients of the first variable from the left. Finally, for $i=5,\dots,10$, we relabel $x_i\rightarrow a_i x_i$ if necessary to ensure that the coefficient of each $x_i$ equals 1 in the first constraint it appears in. 
This yields the LCS
\begin{align}\label{pentagram red}
 x_1  +   x_2 +  x_3  +   x_{4} &= b_1	\notag\\
 x_1 +  x_5  +  x_6  +   x_7 &= b_2	\notag\\
 x_2  +\gamma_5   x_5  +  x_8  +  x_9 &= b_3	\notag\\
 x_3  + \gamma_6  x_6  +   \gamma_8 x_8  +  x_{10} &= b_4	\notag\\
 x_4  +  \gamma_7 x_7  + \gamma_9   x_9  + \gamma_{10}   x_{10} &= b_5,	
\end{align}  
for some $\gamma_i \in \bb{Z}_d-\{0\}$, and (generally different) $b \in \bb{Z}^5_d$.
As before, the reduction preserves quantum and classical satisfiability. Denote the reduced LCS above by $Mx  = b$, where
\begin{align}
M = \left(
\begin{array}{cccccccccc}
 1 & 1 & 1 & 1 & 0 & 0 & 0 & 0 & 0 & 0 \\
 1 & 0 & 0 & 0 & 1 & 1 & 1 & 0 & 0 & 0 \\
 0 & 1 & 0 & 0 &\gamma_5 & 0 & 0 & 1 & 1 & 0 \\
 0 & 0 & 1 & 0 & 0 & \gamma_6 & 0 & \gamma_8 & 0 & 1 \\
 0 & 0 & 0 & 1 & 0 & 0 & \gamma_7 & 0 & \gamma_9 & \gamma_{10} \\
\end{array}
\right).
\end{align}
It is straightforward to verify that
\begin{align}\label{mp rank}
\text{rank}(M) = 
\begin{cases}
4 & \text{if } \gamma_i = -1 \text{ for all } i,\\
5 & \text{otherwise},
\end{cases}
\end{align} 
where the rank is computed modulo $d$.
Thus the system $Mx =b$ is classically unsatisfiable only if $\gamma_i = -1$ for all $i$. 
Again, we find that phase-commutation prohibits the existence of a satisfiability gap in this case.
\begin{theorem}\label{mp thm}
Let $d > 2$ be a prime integer. A LCS of the form in \cref{mp dlc} over $\bb{Z}_d$ has a quantum solution if and only if it is classically satisfiable.
\end{theorem}
\begin{proof}
If a classical solution exists then a quantum solution exists. Conversely, assume that a quantum solution exists, and denote the associated reduced LCS in \cref{pentagram red} by $Mx=b$. Assume that $\gamma = (-1,-1,-1,-1,-1,-1)$, otherwise there is a classical solution and we are done. 
In the solution group $\mc{G}(M,b)$, we have
\begin{align} \label{mp dlc os}
g_1  g_2  g_3 g_4= J^{b_1} \quad
g_1  g_5 & g_6 g_7 = J^{b_2} \quad
g_2  g_5^{-1}  g_8 g_9 = J^{b_3} \nonumber\\
g_3  g_6^{-1}  g_8^{-1}  g_{10}=& J^{b_4} \quad
g_4  g_7^{-1}  g_9^{-1} g_{10}^{-1}= J^{b_5}.
\end{align}
We can deduce a phase-commutation relation by multiplying the last four constraints:
\begin{align}
J^{b_2 + b_3 + b_4 + b_5} &= (g_1 g_5 g_6 g_7 )( g_2  g_5^{-1}  g_8 g_9  ) (g_3  g_6^{-1}  g_8^{-1}  g_{10}) (g_4  g_7^{-1}  g_9^{-1} g_{10}^{-1}) \nonumber\\
&= (g_1 g_6 g_7 )( g_2 g_9  ) (g_3  g_6^{-1} ) (g_4  g_7^{-1}  g_9^{-1} )  \nonumber\\
&= (g_1 g_7 )( g_2 g_9  ) (g_3) (g_4  g_7^{-1}  g_9^{-1} ) \nonumber\\
&= (g_1 g_7 )( g_2 g_9  ) (g_7^{-1}  g_9^{-1})  (g_3 g_4 )  \nonumber\\
&=(g_1 g_7 )g_2 g_7^{-1}   (g_3 g_4) \nonumber\\
&=g_7 ( g_1  g_2 g_4 ) g_7^{-1}  g_3
\end{align}
where in the third and fourth lines we use the facts that $g_6$ commutes with $g_2 g_9$, and $g_3$ commutes with $g_7g_9$, respectively. 
Substituting $g_1 g_2 g_4 = J^{b_1} g_3^{-1}$ reveals the phase-commutation relation
\begin{align}\label{pent pc}
g_7 g_3^{-1} &= J^{-b_1 + b_2 + b_3 + b_4 + b_5} g_3^{-1}  g_7.
\end{align}
Since $d$ is odd, by \Cref{thm: phase commutation} and \cref{pent pc}, a necessary condition for the existence of a quantum solution is that $-b_1 + b_2 + b_3 + b_4 + b_5 = 0 \mod d$.
This implies that the LCS is classically satisfiable: set $x_b = (b_2,b_3,b_4,b_1-b_2-b_3-b_4,0,0,0,0,0,0)$, then $M x_b = (b_1,b_2,b_3,b_4,b_1-b_2- b_3 - b_4) = b$. This shows that the original LCS in \cref{mp dlc} has a classical solution.
\end{proof}
\subsection{Magic square and pentagram modulo odd, composite  $d$}
When $d$ is a composite integer, $\bb{Z}_d$ is not a field, and one cannot define a vector space over $\bb{Z}_d$.
This makes it difficult to generalize the proofs of the previous two subsections to this setting.
In particular, zero-divisors can appear as coefficients in the LCS, making the reductions used in \cref{ms dlc reduced,pentagram red} inapplicable.
Linear algebraic concepts, such as matrix rank and column space, are also invalid in this case, so we cannot use \cref{ms rank,mp rank}.

Nevertheless, if the coefficients in the LCS are restricted to $\{ \pm 1\}$, the arguments used in the prime $d$ case go through almost unchanged. Namely, we can use the same reductions to obtain the LCSs \cref{ms dlc reduced,pentagram red}, and these LCSs satisfy the following.
\begin{lemma}\label{lin alg ms}
If the LCS in \cref{ms dlc reduced} is classically unsatisfiable modulo an odd integer and $\gamma_i \in \{ \pm 1\}$, then $\gamma_4=\gamma_5=\gamma_6 $ and $\gamma_7=\gamma_8=\gamma_9$.
\end{lemma}
\begin{lemma}\label{lin alg mp}
If the LCS in \cref{pentagram red} is classically unsatisfiable modulo an odd integer and $\gamma_i \in \{ \pm 1\}$, then $\gamma = (-1,-1,-1,-1,-1,-1)$.
\end{lemma}

Both lemmas are proven in \Cref{app: proof}. 
We can now immediately generalize \Cref{ms thm,mp thm} to the case of odd $d$ and coefficients in $\{ \pm 1\}$:

\begin{theorem}\label{ms thm odd}
Let $d>1$ be odd. A LCS of the form in \cref{ms dlc} over $\bb{Z}_d$ with coefficients in $\{ \pm 1\}$ has a quantum solution if and only if it is classically satisfiable.
\end{theorem}

\begin{theorem}\label{mp thm odd}
Let $d>1$ be odd. A LCS of the form in \cref{mp dlc} over $\bb{Z}_d$ with coefficients in $\{ \pm 1\}$ has a quantum solution if and only if it is classically satisfiable.
\end{theorem}

The proofs of \Cref{ms thm odd,mp thm odd} are identical to those of \Cref{ms thm,mp thm}, respectively.

\section{Conclusion and remarks}
The existence of a quantum/classical satisfiability gap in a LCS modulo $d$ is deeply related to whether $d$ is even or odd.
Part of this dependence is elucidated by considering generalized Pauli operators.
These operators enter the picture in two ways. First, as shown in \Cref{thm:pauli_assignment}, generalized Pauli operators in odd $d$ cannot be used to demonstrate a satisfiability gap, since every quantum solution consisting of these operators can be reduced to a classical solution.
Second, the canonical commutation relation $AB = \omega^c BA$, obeyed by generalized Pauli operators, is itself sufficient to rule out a satisfiability gap when $d$ is odd. 
Specifically, \Cref{thm: phase commutation} shows that if this relation is embedded in the solution group of a LCS over odd $d$, then quantum solutions do not exist in any dimension.
Both of these facts are in contrast to the case of even $d$; \cref{tab:table} shows a family of LCSs for every even $d$, each of which has a phase-commutation relation, as well as a satisfiability gap which can be demonstrated using generalized Pauli operators in even $d$.
\Cref{thm: phase commutation} can be seen as an obstruction to generalizing the embedding theorem of \cite{slofstra2019tsirelson}, which holds for $d=2$, to odd $d$. The embedding theorem has significant consequences; it plays a key role in proving the separation between the tensor-product and commuting-operator models of two-party quantum correlations.
It is a curious fact that this theorem does not hold if $d$ is odd.

In \Cref{sec: examples} we use the relationship between phase-commutation and satisfiability to prove that certain generalizations of the magic square and pentagram to odd $d$ do not have a satisfiability gap.
The LCSs we consider in these examples can be seen as arising from the incidence matrices of (weighted) graphs. 
Prior work for $d = 2$ characterizes the existence of a gap in terms of planarity of such graphs \cite{arkhipov2012extending}, but this characterization does not carry over to the case of $d > 2$.
The techniques based on phase-commutation, outlined in the examples in \Cref{sec: examples}, may become useful in finding such a generalization.
It is, however, unclear how simple such a characterization might be. 
Planarity of a graph can be checked efficiently, see for example  \cite{boyer2004cutting}, but phase-commutation seems to be a more difficult property to check.

Finally, the results in this paper raise the question of whether there exists a LCS over odd $d$ with a classical/quantum satisfiability gap.
While no such cases are currently known,
an upcoming work by Slofstra and L. Zhang \cite{slofstra2019private} answers this question in the affirmative; there exists a LCS with a  classical/quantum gap modulo $d$ for every prime $d$, with some evidence that the gap is achievable using finite-dimensional quantum solutions.
\\

\noindent \textbf{Acknowledgment:} We thank William Slofstra and Richard Cleve for fruitful discussions and insights, as well as Juani Bermejo-Vega for feedback on an earlier version of the manuscript. This research was undertaken thanks in part to funding from TQT.

\bibliography{lcslib}
\bibliographystyle{unsrt}
\appendix
\section{Proofs}\label{app: proof}
\begin{proof}[Proof of \Cref{thm: char1}]
Suppose $(M,b)$ has a quantum solution $\{ A_i \}$. By \cref{reps of sg}, there is a representation $\phi$ of $\mc{G}(M,b)$ such that $\phi(J) = \omega I$. For $c \neq 0 \mod d$, $\phi(J^c) = \omega^c I  \neq I$, and therefore $J^c \neq e$. Thus the order of $J$ is $d$. Conversely, suppose that $\mc{G}(M,b)$ is such that the order of $J$ is $d$. Define the complex vector space $\mc{H}$ by
\begin{align}
\mc{H} = \Bigg\{ \sum_{g \in \mc{G}} \alpha_g \ket{g} : \norm{\alpha}^2 < \infty \Bigg\}.
\end{align}
Define the operators $\{L_g : g\in \mc{G} \}$ via $L_g \ket{h} = \ket{gh}$. The proof proceeds in two steps. First we list certain properties of the $L_g$. Second, we define a certain subspace of $\mc{H}$ on which the restriction of $L_{g_1}, \dots, L_{g_n}$ is a quantum solution. The relevant properties of the $L_g$ are the following
\begin{enumerate}
\item The eigenvalues of $L_{g_i}$ are $d$th roots of unity:
\begin{align} \label{eq: cond1}
L^d_{g_i} = L_{g_i^d} = L_{e} = I.
\end{align}

\item Let $i, j$ and $k$ be such that $M_{ij}\neq 0$ and $M_{ik} \neq 0$, then, for all $g \in \mc{G}(M,b)$, 
\begin{align}\label{eq: cond2}
L_{g_j} L_{g_k} \ket{g} = \ket{g_j g_k g} = \ket{g_k g_j g} = L_{g_k} L_{g_j} \ket{g}, 
\end{align}
and therefore $[L_{g_j}, L_{g_k} ] = 0$.

\item For any $i \in \{1, \dots, m \}$ and $g \in \mc{G}(M,b)$,
\begin{align}\label{eq: cond3}
L_{g_1}^{M_{i1}} \dots L_{g_n}^{M_{in}}\ket{g} 
&= \ket{g_{1}^{M_{i1}} \dots g_{n}^{M_{in}} g} \nonumber\\
&=\ket{J^{b_i} g} \nonumber\\
&= L_{J}^{b_i} \ket{g}.
\end{align}

\item Define a fiducial vector $\ket{\psi} =d^{-1/2} \sum_{c \in \bb{Z}_d} \omega^c \ket{J^c}.$ Then
\begin{align}\label{eq: cond4}
L_J^{ - b_i} \ket{\psi} =  \sum_{c \in \bb{Z}_d} \omega^c \ket{J^{c-b_i}} = \omega^{b_i} \ket{\psi}.
\end{align}
\end{enumerate}

To construct the quantum solution, first we define the subspace $\mc{H}' \subset \mc{H}$ 
\begin{align}
\mc{H}' \equiv \Bigg\{ E\ket{\psi} : E \in \mathfrak{A} \Bigg\},
\end{align}
where $\mathfrak{A}$ is the algebra generated by $L_{g_1}, \dots, L_{g_n}$ (i.e. all linear combinations of products of the $L_{g_i}$). Note that the powers of $J$ are all distinct by assumption: if $J^a = J^b$ then $J^{a-b} = e$, and hence $a-b = 0 \mod d$.  In particular this implies that $\ket{\psi} \neq 0$, and, as $\ket{\psi} \in \mc{H}'$, $\mc{H}' \neq \{ 0\}$. 

Define the quantum solution $A_i$ as the restriction of $(L_{g_i})^{-1}$ to $\mc{H}'$, i.e. $A_i = L_{g_i}^{-1}|_{\mc{H}'}$.
Let us verify that this satisfies the definition of a quantum solution:
\begin{enumerate}
\item $A_i^d = I$ follows from \cref{eq: cond1},
\item if $M_{ij}\neq 0$ and $M_{ik} \neq 0$ then $[A_j,A_k] = 0$ follows from \cref{eq: cond2},
\item First use \cref{eq: cond4} and the fact that $ L_J^{b_i}$ commutes with every element of $E  \in \mathfrak{A}$ to deduce that $L_J^{-b_i} E \ket{\psi} =  E L_J^{-b_i} \ket{\psi} = \omega^{b_i} E\ket{\psi}$.
From \cref{eq: cond3} we then have
\begin{align}
 \prod_{j} A_{j}^{M_{ij}} =  \prod_{j} L_{g_j}^{-M_{ij}}|_{\mc{H}'} = L_J^{-b_i}|_{\mc{H}'}  = \omega^{b_i} I.
\end{align}
\end{enumerate}
\end{proof}
\begin{proof}[Proof of \cref{lin alg ms}] Due to symmetry under the exchange $(x_4,x_5,x_6) \leftrightarrow (x_7,x_8,x_9)$, it is enough to show that if $(\gamma_4 = \gamma_5 = \gamma_6)$ does not hold then a classical solution exists. Suppose that $\gamma_4 = - \gamma_5$. Then a classical solution is given by
\begin{align}
x_1 &\rightarrow \frac{1}{2} (b_1-\gamma_4 b_2- \gamma_7 b_3+b_4-b_5-b_6)\quad
\quad \;\;\: x_2 \rightarrow \frac{1}{2} (b_1+\gamma_4 b_2+\gamma_7 b_3-b_4+b_5-b_6) \nonumber\\
x_4 &\rightarrow \frac{1}{2} \gamma_4 (-b_1+\gamma_4 b_2-\gamma_7 b_3+b_4+b_5+b_6) \quad
x_5 \rightarrow \frac{1}{2} \gamma_4 (b_1+\gamma_4 b_2+\gamma_7 b_3-b_4-b_5-b_6)  \nonumber\\
x_3 &\rightarrow b_6, \quad 
x_7 \rightarrow b_3, \quad
x_6,x_8, x_9 \rightarrow 0.
\end{align}
Suppose now that $\gamma_4 = \gamma_5$ and $\gamma_5 = - \gamma_6$. Then a classical solution is given by
\begin{align}
x_1 &\rightarrow  \frac{1}{2} \left(b_1 + \gamma_4 b_2 -\gamma_7b_3 +b_4-b_5-b_6\right) \quad\quad\;\;\:
x_3 \rightarrow \frac{1}{2} \left(b_1+ \gamma_4 b_2+\gamma_7 b_3-b_4-b_5+b_6\right) \nonumber\\
x_4 &\rightarrow \frac{1}{2}\gamma_4\left(- b_1 +\gamma_4 b_2-\gamma_7 b_3 + b_4 + b_5 + b_6 \right)\quad
x_6 \rightarrow \frac{1}{2} \gamma_4\left( b_1 +\gamma_4 b_2+ \gamma_7 b_3- b_4- b_5+ b_6\right) \nonumber\\
x_2 &\rightarrow b_5,\quad  x_7 \rightarrow b_3,\quad x_5 ,x_8,x _9 \rightarrow 0.
\end{align}
These solutions are valid modulo $d$ as long as 2 has a multiplicative inverse modulo $d$, which is the case when $d$ is odd.
\end{proof}
\begin{proof}[Proof of \cref{lin alg mp}] 
The following explicit solutions
{\footnotesize
\begin{align}
x\rightarrow \Big(\frac{1}{2} \left(b_1+b_2-b_3-b_4-b_5\right),\frac{1}{2} \left(b_1-b_2+b_3-b_4-b_5\right),b_4,b_5,\frac{1}{2} \left(-b_1+b_2+b_3+b_4+b_5\right),0,0,0,0,0\Big) & \quad \text{if } \gamma_5 =1, \nonumber\\
x\rightarrow \Big(\frac{1}{2} \left(b_1+b_2-b_3-b_4-b_5\right),b_3,\frac{1}{2} \left(b_1-b_2-b_3+b_4-b_5\right),b_5,0,\frac{1}{2} \left(-b_1+b_2+b_3+b_4+b_5\right),0,0,0,0\Big)& \quad \text{if } \gamma_6 =1,\nonumber\\
x\rightarrow \Big( \frac{1}{2} \left(b_1+b_2-b_3-b_4-b_5\right),b_3,b_4,\frac{1}{2} \left(b_1-b_2-b_3-b_4+b_5\right),0,0,\frac{1}{2} \left(-b_1+b_2+b_3+b_4+b_5\right),0,0,0 \Big) & \quad \text{if } \gamma_7 =1,\nonumber\\
x\rightarrow \Big( b_2,\frac{1}{2} \left(b_1-b_2+b_3-b_4-b_5\right),\frac{1}{2} \left(b_1-b_2-b_3+b_4-b_5\right),b_5,0,0,0,\frac{1}{2} \left(-b_1+b_2+b_3+b_4+b_5\right),0,0 \Big) & \quad \text{if } \gamma_8 =1,\nonumber\\
x\rightarrow \Big( b_2,\frac{1}{2} \left(b_1-b_2+b_3-b_4-b_5\right),b_4,\frac{1}{2} \left(b_1-b_2-b_3-b_4+b_5\right),0,0,0,0,\frac{1}{2} \left(-b_1+b_2+b_3+b_4+b_5\right),0 \Big) & \quad \text{if } \gamma_9 =1,\nonumber\\
x\rightarrow \Big( b_2,b_3,\frac{1}{2} \left(b_1-b_2-b_3+b_4-b_5\right),\frac{1}{2} \left(b_1-b_2-b_3-b_4+b_5\right),0,0,0,0,0,\frac{1}{2} \left(-b_1+b_2+b_3+b_4+b_5\right) \Big) & \quad \text{if } \gamma_{10} =1,
\end{align}}%
\noindent prove the statement, provided that $2$ has a multiplicative inverse modulo $d$, which is the case if $d$ is odd.
\end{proof}
\end{document}